\documentclass[10pt,conference]{IEEEtran}
\IEEEoverridecommandlockouts

\usepackage{cite}
\usepackage{amsmath,amssymb,amsfonts}
\usepackage{amsthm}
\usepackage{algorithm}
\usepackage{algorithmic}
\usepackage{graphicx}
\usepackage{textcomp}
\usepackage[table]{xcolor}
\usepackage{xspace}
\usepackage{booktabs}
\usepackage{array}
\usepackage{tikz}
\usetikzlibrary{arrows.meta,positioning,calc,fit,backgrounds}
\usepackage[hidelinks]{hyperref}
\usepackage{listings}
\usepackage{multirow}

\lstdefinestyle{cedarstyle}{
  basicstyle=\ttfamily\scriptsize,
  backgroundcolor=\color{black!5},
  frame=none,
  showstringspaces=false,
  breaklines=true,
}

\def\BibTeX{{\rm B\kern-.05em{\sc i\kern-.025em b}\kern-.08em
    T\kern-.1667em\lower.7ex\hbox{E}\kern-.125emX}}

\newtheorem{definition}{Definition}

\newcolumntype{L}[1]{>{\raggedright\arraybackslash}p{#1}}

\newcommand{\sem}{\mathrm{sem}}
\newcommand{\sys}{\textup{\textsc{AutoCedar}}}

\begin{document}

\title{\sys: An Agentic Framework for\\
Verifier-Guided Access Control Policy Synthesis}

\author{\IEEEauthorblockN{Adarsh Vatsa, Sachi Shome, Yingming Zhou, and William Eiers}
\IEEEauthorblockA{Stevens Institute of Technology\\
\{avatsa, sshome, yzhou136, weiers\}@stevens.edu}}

\maketitle
\begin{abstract}
Large Language Models are increasingly used to turn natural-language
requirements into code. In access control, that shortcut is dangerous: a
generated policy can compile and read correctly while granting access that no
one approved. The difficulty is not only writing policy code. It is fixing what
the requirements mean before code is written, and then checking that the final
policy actually satisfies that intent. We present \sys{}, a verifier-guided
system that first turns natural-language access-control requirements into a
reviewed, checkable target, and then synthesizes Cedar policies against that
target. \sys{} splits schema and policy authoring into small atoms: schema
atoms, which declare vocabulary, and condition atoms, which each state one
requirement in Cedar. A person or a model approves each atom, and
\sys{} calculates the target from them. An LLM then proposes a
candidate policy, the verifier checks it against the target, and each failure
becomes a repair packet that tells the model whether to tighten, loosen, or
expand the policy without changing the target. Because the model's work is
split into small problems, each grounded in reviewed intent and backed by
verifier feedback, end-to-end policy authoring becomes tractable. \sys{}
converges on all 221 tasks of CedarBench, our benchmark of authorization tasks,
each with a given schema and verification plan. 
On 200 cases with independently compiled gold policies, \sys{}’s policy is equivalent to the gold policy in 178 (89.0\%), compared with 132 (66.0\%) for direct generation with the same model. On requests derived from three requirements corpora, \sys{}’s generated schema expresses 270 of 316 requests, and its policy decides 268 of those 270 correctly; direct generation expresses 253 and decides only 109 of them correctly.
\end{abstract}

\section{Introduction}\label{sec:intro}

Access control is ubiquitous across deployed software systems. However, the gap
between natural-language access control requirements and the formal policies that
enforce them remains a persistent source of misconfiguration and vulnerability.
In fact, this gap makes access control the most common source of security
failures: it ranks first in the OWASP Top 10, and it appeared in every
application tested in their survey~\cite{owasp2025brokenaccess}. These failures
do not persist because any single access control rule is hard to write.
Individual rules can each look reasonable and still combine to permit or deny
requests that no author intended. Instead, the failures persist because
correctness is a property of the entire system. A policy is correct only when
its entire set of rules, taken together, permits exactly the requests that it
should and no others.

This problem compounds in modern software development, where AI tools now author
a large and growing share of committed code. Fewer than half of developers say
that they always verify AI-generated code before it ships, and a substantial
share report that reviewing it takes more effort than reviewing human-written
code~\cite{sonar2026trustgap}. Authorization logic is not exempt. More of it is
being generated, faster, by systems that have no internal signal for whether a
rule matches intent. Automated vulnerability discovery makes the resulting gap
visible, but it does not close it. For example, Anthropic's Project Glasswing~\cite{anthropic2026glasswing}
reports that around fifty partners using Claude Mythos have found more than ten
thousand high- or critical-severity vulnerabilities in important software. Similarly, a parallel
scan of over a thousand open-source projects surfaced 6,202 high or critical
findings out of 23,019 candidates. However, the bugs that Anthropic describes in
detail are primarily memory-safety vulnerabilities~\cite{anthropic2026mythos}. The
reason is structural. A memory bug carries its own ground truth: a buffer
overflow trips a sanitizer and crashes, and no statement of intent is required.
In contrast, ``user A can read user B's record'' is a defect only relative to a
policy that the code never states, and fuzzers cannot easily detect such logic
bugs~\cite{anthropic2026mythos,felmetsger2010logic}. Hence, authorization defects often escape
both human review and automated bug-finding, so correctness must be ensured at
authoring time.

Large language models (LLMs) can turn natural-language requirements directly
into policy text~\cite{vatsa2025synth,vatsa2025explore}. However, a model can produce a
syntactically valid, plausible-looking policy that is semantically wrong, and
the model has no internal signal that anything is
off~\cite{huang2024selfcorrect,stechly2025selfverification}. For example, Vatsa et
al.~\cite{vatsa2025explore} report that policies written via roundtrip
generation have semantic equivalence rates, relative to the ground-truth
policies, of 45.8\% for non-reasoning models and 93.7\% for reasoning models.
A natural response is to pair the model with a verifier, which separates
generation from checking~\cite{kambhampati2024llmmodulo}. The correctness guarantees then come from the
verifier, and the model does not have to provide them. We argue that this is
only half of the problem. The requirements themselves are usually incomplete and
leave cases unsaid~\cite{luitel2024completeness}. In a survey of 228 companies, incomplete or hidden
requirements were the most cited problem~\cite{mendezfernandez2017napire}.
No amount of verification against an under-specified,
informal intent can fill that gap.

The Cedar~\cite{cedar2024,disselkoen2024cedar} policy language makes verification more tractable,
since Cedar policies are declarative, typed against an explicit schema, and open
to symbolic comparison. Thus, we can require more than a policy that
looks reasonable. The property that we want is that the set of requests it permits
falls inside a reviewed authorization boundary. This boundary does not exist in
ordinary requirements text~\cite{zave1997dark}. It has to be constructed, reviewed, and
kept fixed while the policy is synthesized.

In this paper, we present \sys{}, which constructs that reviewed boundary by
completing intent rather than merely translating requirements into policy text.
\sys{} does not treat the requirements or the schema as finished. Instead, it surfaces their gaps during authoring. For example, when the
requirements do not say whether a list of allowed cases is complete, \sys{} asks a
reviewer, who is a person or an LLM. When the schema cannot express a requirement,
\sys{} repairs the schema, and the reviewer approves the repair. Both happen
before the policy is written.

Because \sys{} records the completed intent explicitly, it also produces a
durable, updatable log of intent. When a reviewer finds a mistake or a requirement
changes, only the affected parts of the log are written and
reviewed again, and a fresh, verified policy is synthesized from the log. Hence, \sys{} is a substrate for ongoing policy
maintenance.

\sys{} takes as input a natural-language document of access control
requirements and, if one exists, the application's Cedar schema. It then
builds three artifacts in order. First, it builds the schema from the
requirements, or it extends the given schema with the attributes and actions
that the requirements need. Second, it builds the reviewed boundary, which we call the
target. The target comes from an intent map, which records which requests each
sentence constrains, and from small atoms that the reviewer approves one at a time.
Each atom traces to a source sentence. Third, it synthesizes
the policy against this fixed target, guided by the failures that Cedar's
validator and symbolic compiler report. Thus, the returned policy is verified against a schema
and a target that the reviewer approved.

\textbf{Contributions.} We make the following contributions.
\begin{enumerate}
  \item \sys{}, a verifier-guided system that calculates the target from reviewed
  artifacts before any policy is written, and traces every atom to its source
  sentences (Sections~\ref{sec:schema}, \ref{sec:hitl}, and~\ref{sec:maintain}).
  \item A synthesis loop that returns each failure of Cedar's validator and symbolic
  compiler to the model as a repair packet with a direction of repair and any
  counterexample, while the target stays fixed (Section~\ref{sec:mechanisms}).
  \item CedarBench, a benchmark of 221 executable authorization tasks, each with a
  Cedar schema and a verification plan, which states its target formally
  (Section~\ref{sec:cedarbench}). Given these plans, \sys{} synthesizes a policy
  that satisfies the plan of every task.
  \item An evaluation of how well \sys{} captures intent (Section~V). On 200 cases with gold policies, the policy of \sys{} is equivalent to the gold policy in 89.0\% of cases, versus 66.0\% for direct generation with the same model. On requests derived from three requirements corpora in healthcare, education, and conference management~\cite{xiao2012text2policy,slankas2013acre,slankas2014rede}, \sys{}’s schema expresses 85.4\% of requests and its policy decides 99.3\% of those correctly, versus 80.1\% and 43.1\% for direct generation.
\end{enumerate}
The implementation, benchmarks, and reproduction materials are available at
\url{https://github.com/neselab/cedar-synthesis-engine}. \sys{} is also
distributed as a Python package and as a
plugin for coding agents such as Claude Code and Codex.

\section{Background and Preliminaries}
\label{sec:background}

\subsection{Cedar Policy Language}
Cedar is a language for writing authorization policies and for making
authorization decisions based on those policies~\cite{cedar2024,cedarpolicy}.

\textbf{Schema.} A Cedar schema $\Sigma$ fixes the vocabulary that policies may
reference. It defines the entity types of the application. Policies refer to
\textit{principals} (who is doing the access), \textit{actions} (what is being
performed), \textit{resources} (what is being accessed), and \textit{context}
(environment information). At authoring time, Cedar's validator checks that every entity,
attribute, and action that a policy references is consistent with $\Sigma$.


\textbf{Policies.} A Cedar policy $P$ is a finite set of \texttt{permit} and
\texttt{forbid} rules. Each rule specifies an \textit{effect}, a \textit{scope},
and optional \textit{conditions} as follows:
\[
\texttt{effect}\;(\textit{principal}, \textit{action}, \textit{resource}) [\texttt{when}\;\varphi_1][\texttt{unless}\;\varphi_2]
\]
where $\texttt{effect} \in \{\texttt{permit}, \texttt{forbid}\}$, the scope
restricts the principals, actions, and resources to which the rule applies,
and $\varphi_1$ and $\varphi_2$ are Boolean conditions over entity attributes
drawn from $\Sigma$.



\subsection{Terminology}

\textbf{Requests.} We write a Cedar authorization
request as a tuple $\rho=(p,\alpha,o,\gamma)$ of a principal, an action, a
resource, and a context. For a fixed schema $\Sigma$, the declarations of
$\Sigma$ induce the request universe $R_\Sigma$. $R_\Sigma$ is the domain over
which every comparison and certification claim is quantified.

\textbf{Permit-set denotation.} Cedar uses deny-overrides semantics: a
request is permitted exactly when at least one matching \texttt{permit} applies
and no matching \texttt{forbid} applies. We write
$
\sem_\Sigma(P)=
\{\rho\in R_\Sigma \mid \mathsf{permit}_\Sigma(P,\rho)\}
$
for the permit-set denotation of $P$, and we write $\sem(P)$ when $\Sigma$ is
fixed. Two different Cedar encodings are equivalent for synthesis when they
permit exactly the same requests.

\textbf{Symbolic policy comparison.} Cedar comes with a symbolic
compiler~\cite{cedaranalysis2025}, \texttt{cedar symcc}, which translates Cedar
policies into formulae of satisfiability modulo theories (SMT). For policies
$P_1$ and $P_2$ in the supported Cedar fragment, it can decide inclusion
$\sem(P_1)\subseteq\sem(P_2)$, equivalence $\sem(P_1)=\sem(P_2)$, and
satisfiability or emptiness properties. In contrast, testing compares policies
only over a finite test suite. If an inclusion check fails, the compiler returns a
\emph{counterexample}, a concrete request in the relevant set difference. This approach follows the
same lineage as SMT-based policy analyzers such as Zelkova~\cite{zelkova2018} and
XACML analyzers~\cite{fisler2005margrave,hughes2008sat,turkmen2017xacml}.
The property that we use is that many requirements can themselves be encoded as
Cedar policies.


\begin{figure*}[t]
\centering
\includegraphics[width=\textwidth]{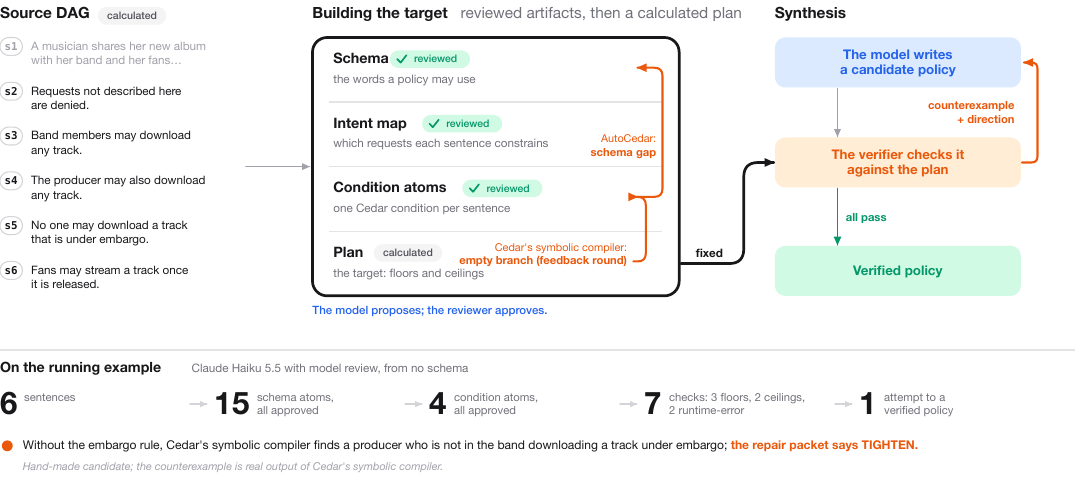}
\caption{\sys{} on the running example. The target is calculated from reviewed artifacts
and fixed before the model writes the policy.
Orange marks each problem that is sent back to an earlier step, and the
strip below gives the counts from the run.}
\label{fig:arch}
\end{figure*}

\section{\sys{}}
\label{sec:framework}

\sys{} takes as input a natural language document $D$ consisting of access
control sentences, with the goal of synthesizing a Cedar policy which implements
the requirements. We use one example throughout this section. Suppose a musician shares
her new album through a web app. Band members and the producer may download
tracks; no one may download a track that is under embargo; fans may stream a track
once it is released; and every other request is denied. The document has six
sentences, $s_1$ to $s_6$ (Figure~\ref{fig:arch}).

Before it writes any policy, \sys{} builds a \emph{target}: a set of short Cedar
policies that \sys{} uses only to check the final policy. Each of them covers one
\emph{scope}, that is, one principal type, one action, and one resource type, such
as users who download tracks. Each is either a floor or a ceiling. A \emph{floor}
is a policy whose requests the final policy must allow. For example, the floor for
band members allows a member of the album's band to download a track that is not
under embargo. A \emph{ceiling} is a policy that bounds what the final policy may
allow: the final policy must deny every request of the scope that the ceiling
denies. For example, the ceiling for downloads allows a download only by a member
of the album's band or by the producer, and only of a track that is not under
embargo. Requests that $D$ leaves open lie outside every floor and inside every
ceiling, and the final policy may allow or deny them.

\sys{} uses an LLM, which we call \emph{the model}, to propose the final policy and
three kinds of small artifacts. \emph{Schema atoms} declare the vocabulary of the
schema, an \emph{intent map} records which requests each sentence constrains, and
\emph{condition atoms} state each sentence in Cedar. A
\emph{reviewer} approves or rejects each of these artifacts. In \emph{human
review}, the reviewer is a person. In \emph{model review}, an LLM judges each
artifact against $D$ in place of a person~\cite{zheng2023judge}. \sys{} can assign a separate LLM to each
role, including the reviewer. \sys{} then calculates the floors and
ceilings from the approved artifacts and fixes them, so the reviewer never writes
or judges them directly. After that, the model writes the final policy against
them (Figure~\ref{fig:arch}). Cedar's validator and symbolic compiler check the
schema and each condition atom as they are built. We call the two tools together the \emph{verifier}. The key
insight is that the reviewer approves the target one small artifact at a time,
while the verifier checks the final policy against it.

We ran \sys{} on the example with Claude Haiku~5.5 and model review, starting from
no schema. Every artifact we show comes from this run, unless we state that we
built a case by hand.

The floors are themselves Cedar policies, so a reader may ask why \sys{} does not
simply join them into the final policy. Every floor is a \texttt{permit}, and Cedar
permits never conflict, so the joined floors form a valid policy that passes any
plan that \sys{} calculates. Real conflicts between requirements, where one
sentence allows a request that another denies, are caught while the target is
built (Section~\ref{sec:mechanisms}). Such a conflict goes back for review instead
of being resolved silently by Cedar's deny-overrides rule. However, a plan
that \sys{} does not calculate, such as the plan that each CedarBench task supplies
(Section~\ref{sec:cedarbench}), can require the policy to allow requests that no
floor contains. The floors alone then fail, and only a search against the verifier
finds a passing policy. The model performs that search.

\begin{figure*}[t]
\centering
\includegraphics[width=\textwidth]{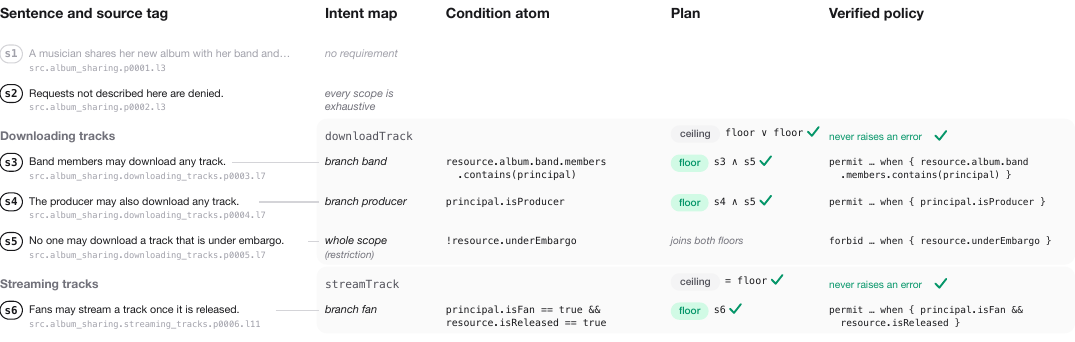}
\caption{The running example, one row per sentence, from source tag to verified policy.
The restriction $s_5$ joins both floors of \texttt{downloadTrack}, and the policy
states it once, as a \texttt{forbid} rule.}
\label{fig:example}
\end{figure*}

\subsection{Schema Construction}
\label{sec:schema}

We first describe how \sys{} indexes the requirements and builds the schema. We
start by indexing the document $D$ without the model: $D$ is split at headings,
blank lines, and sentence boundaries, and each sentence becomes a node of a
directed acyclic graph (DAG), the \emph{source DAG}, in which each section points
to its sentences and each sentence points to the next one. Each node receives a
\emph{source tag} built from its two innermost headings, its position in the
document, and the line on which its paragraph starts (Figure~\ref{fig:example}).
The same document always yields the same tags, and every later artifact cites the
tags of its sentences (Section~\ref{sec:maintain}).

Next, \sys{} builds the Cedar schema $\Sigma$ (Section~\ref{sec:background}). \sys{}
groups the sentences of each section into a \emph{schema packet}, and the model
proposes schema atoms from each packet. A schema atom declares one entity
type, one attribute, one action, or one common type, which is a named type that
other declarations reuse~\cite{cedarpolicy}. The reviewer approves or rejects each
schema atom. \sys{} compiles the approved atoms into $\Sigma$, and Cedar's
validator checks the result. When the application's schema is given, \sys{} uses
it as $\Sigma$ and skips this step.

This review matters because every later check uses only the vocabulary of
$\Sigma$. For example, consider the band member in $s_3$. One reading gives each
user a role string, and the atom for $s_3$ becomes
\texttt{principal.role == "band member"}. Under this reading, a member of any band
may download the album. In our run on the example, the model instead proposed three attributes
whose values are entities: each track has an album, each album has a band, and
each band has a set of members. The atom for $s_3$ then becomes
\texttt{resource.album.band.\allowbreak members.\allowbreak contains(principal)},
which admits only the members of this album's band. No later check can tell that
the first reading is wrong, so it must be caught when the schema is reviewed.

However, the schema can still lack a part that a later condition needs. When the
model writes a condition atom, it therefore lists the parts of the schema that the
atom uses, such as an attribute or an action. A listed part that $\Sigma$ lacks is
a \emph{schema gap}. If the model leaves a part out of its list, Cedar's validator
still rejects the atom. Before the reviewer sees an atom with a schema gap, the
model proposes the missing schema atoms, which are reviewed like any other, and the
atom is checked again against the extended schema. We call this the \emph{schema gap loop} (Figure~\ref{fig:arch}).

\subsection{Target Construction}
\label{sec:hitl}

A sentence on its own may be ambiguous~\cite{ezzini2022anaphoric}. For example, ``The producer may also
download any track'' ($s_4$) can be read in two ways. It can add a second way to
download, so that the producer may download without being a band member. It can
also add a condition to the first way, so that only a band member who is also the
producer may download. The word ``also'' signals the first reading. Similarly, ``No one may download a track that is under embargo''
constrains every way to download, including the producer's.

We make these readings explicit in the intent map, an AND-OR graph~\cite{vanlamsweerde2001goal} that
follows the blueprint of LEAP~\cite{kung2026leap}, an agentic framework for formal
mathematics. Formally, an intent map is a tuple $M=(Q,B,N)$. $Q$ is a set of scopes, such as (\texttt{User},
\texttt{downloadTrack}, \texttt{Track}). We also write $q$ for the set of all
requests of a scope $q\in Q$. $B$ is a set of \emph{branches}. Each branch $b\in B$
is one way to allow the requests of one scope, such as \emph{band} on
\texttt{downloadTrack}. $N$ is a set of \emph{conditions}. Each condition $n\in N$
is a sentence of $D$ that a request either meets or does not meet. Each condition
belongs to one or more branches or \emph{restricts} one or more whole scopes. $M$
allows a request of scope $q$ if and only if the request meets every restriction
of $q$ and every condition of at least one branch of $q$. The remaining sentences
of $D$ are rules for interpretation, such as $s_2$, or carry no requirement, such
as $s_1$. In the example, the scope \texttt{downloadTrack} has two branches:
\emph{band}, with condition $s_3$, and \emph{producer}, with condition $s_4$. The
condition $s_5$ restricts the whole scope (Figure~\ref{fig:example}).

$M$ has a scope only for the actions that $D$ allows or restricts, so requests of
other actions are not checked. Cedar denies such a request unless a rule of the
final policy allows it, and no check of the target rejects such a rule. A scope is \emph{exhaustive} if $D$ lists every way
to allow it, and \emph{open} otherwise. If $D$ does not state whether a list is
complete, \sys{} asks the reviewer once per scope. In the example, $s_2$ makes both
scopes exhaustive.

The model proposes $M$ from the whole of $D$, and it is prompted to apply a
sentence only to the actions that its section or its own text names. Before
review, \sys{} checks that $M$ is well formed, for example that every branch has a
condition. A reviewer who rejects $M$ states each objection. \sys{} applies
each objection that names a concrete edit, such as moving a condition to another
branch, and returns the rest to the model. After three rounds of objections, a
run in human review stops and reports the open objections. In model review, a
reviewer from a model family other than the first reviewer's then judges them. If
it upholds an objection that \sys{} cannot apply, the run stops. Otherwise, \sys{}
applies the upheld edits and accepts $M$.

For each condition $n$ and each scope $q$ that $n$ constrains, the model proposes a
\emph{condition atom}: one Cedar \texttt{permit} on $q$ whose \texttt{when} clause encodes $n$
alone. For example, the condition atom for $s_4$ is \texttt{principal.isProducer}.
The model works from a \emph{condition packet} with $n$, $q$, the other conditions
of the branches of $n$, and the atoms already approved for $q$. After the
checks of Section~\ref{sec:mechanisms} pass, the reviewer sees the atom
in plain English next to its source sentence and approves or rejects it. 
Synthesis starts only when every condition of $M$ has an approved
atom on every scope that it constrains. This coverage is procedural: it does not show that an atom is right. The review of
$M$ decides which sentences count as requirements, and the review of each atom
decides whether it captures its sentence.

\textbf{Verification plan.}\label{sec:sandwich} A \emph{verification plan} states the
target formally. It is a tuple $\Pi=(\mathcal{C},\mathcal{F},\mathcal{G})$ of ceilings, floors, and
\emph{workflow slices}. Each element $X$ is a Cedar policy over $\Sigma$ on one
scope $\mathit{sc}(X)$. A workflow slice $G_k$ allows requests of which the policy
must allow at least one. \sys{} never calculates workflow slices, because each
floor that it calculates contains at least one request
(Section~\ref{sec:mechanisms}) and thus already requires the policy to allow a
request of its branch. Workflow slices occur only in plans that \sys{} does not calculate, such as the plan
that each CedarBench task supplies (Section~\ref{sec:cedarbench}). We call these
\emph{given} plans.

\begin{definition}[Plan satisfaction]
\label{def:plan-satisfaction}
A candidate policy $P$ satisfies $\Pi$, written $P\models\Pi$, iff
\[
\begin{array}{ll}
\forall C_i\in\mathcal{C}. & \sem(P)\cap\mathit{sc}(C_i)\subseteq \sem(C_i),\\
\forall F_j\in\mathcal{F}. & \sem(F_j)\subseteq \sem(P),\\
\forall G_k\in\mathcal{G}. & \sem(P)\cap \sem(G_k)\neq\emptyset.
\end{array}
\]
\end{definition}

Each clause of Definition~\ref{def:plan-satisfaction} is one query to Cedar's
symbolic compiler on the scope of its check. For a scope $q$, let $L_q$ be the set of
requests that some floor on $q$ allows, and $U_q$ the set of requests of $q$ that
every ceiling on $q$ allows, with $U_q=q$ if $q$ has no ceiling. Then $P\models\Pi$ if and only if
$L_q\subseteq\sem(P)\cap q\subseteq U_q$ for every scope $q$, and $P$ allows a
request of every workflow slice.

The plan is calculated with three rules. (1) For each branch, the floor is one
\texttt{permit} whose \texttt{when} clause is the conjunction of the atoms of the branch and
of the restrictions on its scope. (2) For each exhaustive scope, the ceiling is the
disjunction of its floors. Hence, $L_q=U_q$, and the policy must allow exactly the
requests that the approved atoms describe. (3) For each open scope, every
restriction becomes a ceiling, so an open scope without restrictions has no
ceiling. In the example, rule (1) gives three floors, and the two floors of
\texttt{downloadTrack} both contain \texttt{!resource.underEmbargo} from $s_5$.
Rule (2) gives one ceiling for each scope (Figure~\ref{fig:example}). We also
require that the policy raise no runtime error on the requests of each scope, one
more query per scope. The example therefore has seven checks: three floors,
two ceilings, and two runtime-error checks.

Note that the floors of a calculated plan, taken together, satisfy the plan. On an
exhaustive scope, rule (2) makes the ceiling their disjunction. On an open scope,
every floor contains every restriction, so it lies inside every ceiling. Before synthesis, \sys{} also checks that
no floor raises a runtime error. For a given plan, in contrast, the floors alone
need not meet its workflow slices.

\subsection{Feedback from Cedar's Validator and Symbolic Compiler}
\label{sec:mechanisms}

Every step of target construction can fail, and so can every candidate policy
(Table~\ref{tab:mechanisms}). We say that a
change \emph{tightens} a policy if the policy then allows fewer requests, and that
it \emph{loosens} the policy if the policy then allows more. A generic rejection, which says only that a check failed, is not enough, because
it does not tell the model which way to change the policy. The model can then
loosen a policy that already allows too much, tighten one that already allows too
little, or rewrite the policy when the fault lies in the target.

\begin{table}[t]
\centering
\caption{The main checks, what runs them, and where a failure goes}
\label{tab:mechanisms}
\footnotesize
\setlength{\tabcolsep}{3pt}
\begin{tabular}{@{}L{0.41\columnwidth}L{0.19\columnwidth}L{0.34\columnwidth}@{}}
\toprule
Check & Tool & A failure goes to \\
\midrule
\multicolumn{3}{@{}l}{\emph{While the target is built}}\\
Schema compiles & validator & the model, which fixes the schema \\
Atom is one \texttt{permit} on its scope & \sys{} & the repair planner \\
Atom has no schema gap & \sys{}, validator & schema construction \\
Atom allows some but not every request, without errors & symbolic compiler & the repair planner \\
A branch's conditions can hold together & symbolic compiler & the conflicting conditions \\
Each floor never raises an error & symbolic compiler & the floor's conditions \\
\midrule
\multicolumn{3}{@{}l}{\emph{While the policy is written}}\\
Candidate parses and validates & validator & the model, a local repair \\
Candidate stays inside each ceiling & symbolic compiler & the model, \textsc{tighten} \\
Candidate contains each floor & symbolic compiler & the model, \textsc{loosen} \\
Candidate meets each workflow slice & symbolic compiler & the model, \textsc{expand} \\
Candidate never raises an error & symbolic compiler & the model, a local repair \\
The same failures repeat & history & the model, told it is cycling \\
\bottomrule
\end{tabular}
\end{table}

A failed check while the target is built means that the target is wrong, so the
failure goes back to the step that produced it. Each condition atom passes seven
checks before a reviewer sees it (Figure~\ref{fig:sandwich}a). Cedar's validator
checks that the atom type-checks against $\Sigma$. Our own checks confirm that the
atom is one \texttt{permit} on its scope and that it never compares entities of
different types, which Cedar's validator allows. Cedar's symbolic compiler then
checks that the atom allows some request, does not allow every request of its
scope, and raises no runtime error. A seventh check applies only to a compound
requirement. For a disjointness requirement, Cedar's symbolic compiler confirms
that the atom never matches a request that its target matches. For a rate limit,
the check confirms only that the encoding names the counter and the threshold. A rejected atom goes to a repair planner, a single model call that decides
whether to rewrite the atom, rewrite an earlier atom, add schema support, drop the
atom, or ask the reviewer. Note that these checks detect only an atom that cannot
be correct.

Some failures appear only when atoms are combined. Before synthesis, Cedar's
symbolic compiler checks that each floor contains at least one request. If the
conditions of a branch contradict one another, \sys{} deletes conditions one at a
time to find a minimal contradicting set, in which removing any one condition makes
the rest consistent. In an optional \emph{feedback round}, this set goes back to the model,
and the other atoms of the branch keep their approval. Without feedback
rounds, the run stops before synthesis and names the set. We call this the
\emph{empty branch loop} (Figure~\ref{fig:arch}).

At attempt $t$, the model writes a candidate policy $P_t$ against the fixed target.
It receives the schema and a plain-English summary of each floor and ceiling. The
Cedar text of a floor or ceiling reaches the model only after its check fails, so
the model must find a passing policy itself.
Cedar's validator checks $P_t$, and Cedar's symbolic compiler then runs one query
for each floor, ceiling, workflow slice, and runtime-error check. If every check
passes, \sys{} returns $P_t$ together with $\Sigma$, $\Pi$, and a trace. Otherwise,
our \emph{signal layer} turns the failures into a \emph{repair packet} for attempt
$t+1$. The packet names the failed check, the direction of repair, and the
counterexample, if any. The direction follows from the kind of check that fails
(Table~\ref{tab:mechanisms}). If a ceiling fails, $P_t$ allows a request outside
the ceiling, and the direction is \textsc{tighten}. If a floor fails, $P_t$ denies
a request inside the floor, and the direction is \textsc{loosen}. If a workflow
slice of a given plan fails, $P_t$ denies every request in the slice, and the
direction is \textsc{expand}. Otherwise, $P_t$ does not parse, does not validate,
or raises a runtime error, and the model repairs $P_t$ locally.
Note that \textsc{loosen} and \textsc{expand} both make $P_t$ more permissive.
However, a failed floor yields one request that $P_t$ must allow, whereas a failed
workflow slice yields no single request.

In the example, the first candidate passed all seven checks of the plan. Its rules
are in the last column of Figure~\ref{fig:example}. To show a failure, we removed
its \texttt{forbid} rule by hand. Cedar's symbolic compiler then returned a
counterexample in which the producer, who is not a band member, downloads a track
under embargo. The candidate allows this request, and the ceiling of
\texttt{downloadTrack} does not (Figure~\ref{fig:sandwich}b). The signal layer
built the repair packet without a model call. The packet begins as follows.

\noindent\begin{minipage}{\columnwidth}
\begin{lstlisting}[style=cedarstyle]
ceiling__scope_download_track (implies): FAIL
  Your policy is MORE permissive than the ceiling.
  It allows something the ceiling forbids.
  Tighten conditions.
  ... The target is fixed during synthesis; a
  suspected target conflict must be returned for
  review, never silently removed.
\end{lstlisting}
\end{minipage}
The rest of the packet contains the ceiling in Cedar, the candidate's rules for
\texttt{downloadTrack}, and the counterexample. Failures of this kind also occur in
real runs. In a run on the CyberChair requirements, the first candidate denied a
reviewer's request to read a review that the reviewer had submitted. The floor for
this request failed, the repair packet gave the direction \textsc{loosen}, and the
second candidate passed all 79 checks.

\begin{figure}[t]
\centering
\includegraphics[width=\columnwidth]{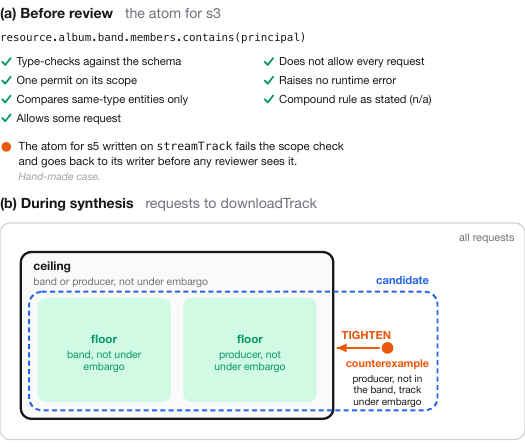}
\caption{Feedback on the running example. (a) The seven checks before review, on the atom for
$s_3$. (b) A hand-made candidate without the
embargo rule fails the \texttt{downloadTrack} ceiling. The counterexample is the actual output of Cedar's
symbolic compiler.}
\label{fig:sandwich}
\end{figure}

The signal layer selects, orders, and compresses evidence. It cannot change $\Pi$,
invent a counterexample, reverse a direction, or add a requirement. It also records
the history of the search, such as a candidate that repeats an earlier one, and
tells the model when the search cycles.

\textbf{Guarantee.} If \sys{} returns a policy $P$ for a plan $\Pi$, then
$P$ has passed Cedar's validator, and Cedar's symbolic compiler has certified every
clause of Definition~\ref{def:plan-satisfaction} and every runtime-error check. The
policy matches the intent of the organization only to the extent that the plan
does. If the model finds no policy within its budget, by default 20 attempts,
\sys{} returns no policy and reports how many checks still fail.

\subsection{A Maintainable Target}
\label{sec:maintain}

Requirements change, and reviewers find mistakes after review. For a policy
written in one pass, either case requires reading the whole policy again, because
nothing records which rule came from which sentence. A new pass can also change
rules that no one asked to change.

During a run, the source DAG of Section~\ref{sec:schema} grows into a single graph, the log of intent.
Each sentence gains edges to the packets, atoms, and schema gaps that come from
it, and the intent map adds edges from conditions to branches and scopes. Hence, a
correction reaches only the artifacts that depend on it. A correction starts from a \emph{finding}, which names the atoms that it concerns. A
finding comes from a reviewer or, in a feedback round, from the empty branch loop
or from an audit of the final policy by an LLM of a family other than the model's. The
approval of the named atoms is withdrawn, and only their sentences go back to the
model, while all other approvals remain. The plan is then recalculated from $M$ and the approved atoms,
and the new policy is checked against the whole new plan.

For example, suppose the reviewer decides that $s_4$ covers only the producer of
this album (Figure~\ref{fig:maintain}). Only $s_4$ goes back. Its new atom needs
\texttt{Album.producer}, which the schema gap loop adds through review, and the
atoms of $s_3$, $s_5$, and $s_6$ keep their approval.

\begin{figure}[t]
\centering
\includegraphics[width=\columnwidth]{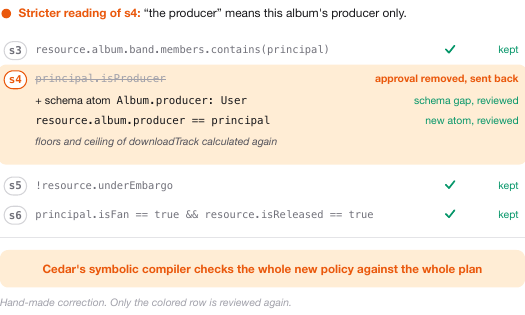}
\caption{A hand-made correction sends back one row.}
\label{fig:maintain}
\end{figure}

After synthesis, the trace $T$ links each sentence to its atoms and checks, and to
the policy rules that cover them, so a reader can look up why a rule exists. $T$ is an
aid for audit and does not prove that a rule came from a sentence. Since \sys{}
currently resumes a run only on the document that it indexed, a change in the
requirements enters a run as a finding on the same document.

\section{CedarBench: A Cedar Benchmark}\label{sec:cedarbench}
CedarBench measures whether a model, guided by verifier feedback, reaches a
policy that satisfies a fixed target. A task asks for a policy that, as a whole,
validates under the task's schema and satisfies the task's plan
(Definition~\ref{def:plan-satisfaction}).

CedarBench contains 221 tasks. Seventy-nine are mutation
tasks, each of which changes one of eight base policies from the official Cedar
examples, for instance by adding a role or removing a constraint. They cover
GitHub-style repositories,
document clouds, hotel chains, sales organizations, streaming services,
tag-and-role systems, tax-preparation workflows, and clinical data access. The
other 142 are hand-authored real-world tasks covering scoped tokens,
delegation, revocation, time windows, break-glass access, role intersections,
and large policy stores. Every task is
executable: it has a prose specification, a Cedar schema, a Python verification
plan, and Cedar reference policies. The reference policies are the floors,
ceilings, and workflow slices of each task's plan, written in Cedar, and they are
not meant for deployment. On CedarBench, \sys{} starts from the given schema and
plan, so only the synthesis loop of Section~\ref{sec:mechanisms} runs. On the
real-world corpora, no schema exists, so \sys{} and direct generation each write
their own schema (Section~\ref{sec:evaluation}).


\section{Experiments}
\label{sec:evaluation}
We consider the following research questions:

\noindent\textbf{RQ1 Synthesis effectiveness}: To what extent can \sys{} synthesize Cedar policies that satisfy formal verification targets, and at what computational cost? \\
\textbf{RQ2 Signal-layer contribution}: How does verifier-to-model feedback contribute to \sys{} synthesis performance? \\
\textbf{RQ3 Verification validity}: Does satisfying the verification plan produced by \sys{} reliably predict policy correctness with respect to NL intent? \\
\textbf{RQ4 Baseline comparison}: How does \sys{} policy quality compare to direct LLM-based policy generation?

\subsection{Experimental Setup}

\textbf{Datasets.} Our evaluation uses the CedarBench benchmark discussed in Section~\ref{sec:cedarbench}, three external real-world natural-language access-control requirements corpora, and a gold-reference dataset of 200 cases. 

The external corpora are the IBM course registration, CyberChair, and iTrust-for-Text2Policy document folds released with RAGent~\cite{jayasundara2024ragent}. RAGent uses a BERT-based identification module to distinguish natural-language access-control-policy (NLACP) sentences from other requirements~\cite{narouei2017identification}. We use the  subsets classified as NLACP sentences in the RAGent release, containing 135 IBM, 116 CyberChair, and 368 iTrust sentences.

Existing natural-language corpora do not provide gold Cedar policies and schemas for independently assessing whether generated verification plans and policies by \sys{} faithfully capture the intended authorization rules. We therefore constructed a gold-reference dataset containing 200 cases spanning 43 application domains. We began with manually designed domain skeletons, expanded their actions and authorization rules with LLM assistance, and manually reviewed the resulting definitions. We then sampled combinations of the reviewed rules and conditions and assigned concrete parameter values to form structured specifications. A deterministic compiler translated each specification, together with its typed domain definitions, into a gold Cedar policy and schema.\footnote{The construction procedure, compiler rules, and instructions for reproducing the gold policies and schemas are provided in the artifact README: \url{https://github.com/neselab/cedar-synthesis-engine/blob/main/dataset/gold-reference-dataset/GOLD_REFERENCE_DATASET_README.md}.} Finally, an LLM rewrote each specification as an organizational requirement in natural language. Each case pairs a natural-language requirement with a gold Cedar policy and schema.

\textbf{Metrics.} We record the model configuration, iteration budget, convergence status, final loss, wall-clock time, token usage, and estimated API cost. Each CedarBench run has a budget of 20 iterations. RQ1 evaluates \sys{} with GPT-5.5 low and Haiku~4.5, and the RQ2 ablation uses GPT-5.5 low. The validation experiment compares direct generation using GPT-5.5 and Haiku~4.5 with \sys{}. Neither direct baseline validates on the real-world corpora (Section~\ref{sec:artifact-validation}), so in the request and gold-reference experiments both methods use GPT-6-Luna and differ only in the workflow. We use the Cedar CLI and \texttt{cedar symcc} with the cvc5 SMT solver~\cite{barbosa2022cvc5}.


\subsection{RQ1: Synthesis Effectiveness}
Table~\ref{tab:rq1-synthesis} reports the results on all 221 CedarBench tasks.
Both models converge on every task. GPT-5.5 low takes fewer repair iterations and
fewer tokens per task. Haiku~4.5 takes more iterations, but is faster and cheaper
in this run. \textbf{Thus, \sys{} synthesizes a passing policy for all 221 tasks with both models, with a mean of under three iterations and under five cents per task.}


\begin{table}[t]
\centering
\caption{CedarBench synthesis results for \sys{} on all 221 tasks. Iterations,
time, tokens, and cost are means per task.}
\label{tab:rq1-synthesis}
\footnotesize
\setlength{\tabcolsep}{4pt}
\renewcommand{\arraystretch}{1.1}
\begin{tabular*}{\columnwidth}{@{\extracolsep{\fill}}lrrrrr@{}}
\toprule
Model & Converged & Iterations & Time (s) & Tokens & Cost (\$) \\
\midrule
GPT-5.5 low & 221/221 & 1.67 & 24.59 & 9,084  & 0.0405 \\
Haiku 4.5   & 221/221 & 2.51 & 8.93  & 12,397 & 0.0185 \\
\bottomrule
\end{tabular*}
\end{table}

\subsection{RQ2: Signal-Layer Contribution}
We run the ablation on a fixed 100-task CedarBench slice and hold fixed the
model (GPT-5.5 low), the tasks, Cedar's validator, the schemas, the verification
plans, and the reference policies. The four conditions differ in what
the model sees: the schema and prose only, the schema and the verification plan in
one shot, native verifier feedback in an iterative loop, or the full signal stack
of \sys{}.

\begin{table}[t]
\centering
\caption{Signal-layer ablation on 100 CedarBench tasks with GPT-5.5 low. Loss is
the mean number of verifier violations at the end of a run. Mean iter.
is averaged over completed tasks, so it can be fractional. Cost is the total for
the run.}
\label{tab:signal-ablation}
\footnotesize
\setlength{\tabcolsep}{2pt}
\renewcommand{\arraystretch}{1.1}
\begin{tabular*}{\columnwidth}{@{\extracolsep{\fill}}lrrrr@{}}
\toprule
Feedback & Converged & Mean loss & Mean iter. & Cost (\$) \\
\midrule
Schema only               & 85           & 15.58         & 1.00 & 1.42 \\
Schema + plan, one shot   & 84           & 0.27          & 1.00 & 1.62 \\
Native verifier loop      & 94           & 0.10          & 2.37 & 5.73 \\
Full signal stack (\sys{}) & \textbf{100} & \textbf{0.00} & 2.20 & 5.83 \\
\bottomrule
\end{tabular*}
\end{table}

Table~\ref{tab:signal-ablation} shows that the number of converged tasks alone
hides how far the failed policies are from the target. The two
one-shot conditions converge on 85 and 84 tasks, but their mean losses differ
sharply: 15.58 versus 0.27. Native verifier feedback converges on 94 tasks,
and the full signal stack closes the last six failures at similar cost. \textbf{Thus, the verification plan brings the model close to a passing
policy, and the repair packets of the signal layer close the remaining
failures.}

\subsection{RQ3 \& RQ4: Policy Correctness and Baseline Comparison}
We answer RQ3 and RQ4 with three kinds of evidence: artifact validation, semantic correctness, and human alignment. The expected request decisions and the gold policies are constructed without consulting the outputs of either system, so they do not favor either approach.

\begin{table}[t]
\centering
\caption{Validation pass rates. Direct generation receives only the natural-language
intent and writes both the schema and the policy. On the real-world corpora,
\sys{} also writes its own schema. On CedarBench, it starts from the given schema
and plan.}
\label{tab:validation-results}
\footnotesize
\setlength{\tabcolsep}{3pt}
\renewcommand{\arraystretch}{1.1}
\begin{tabular*}{\columnwidth}{@{\extracolsep{\fill}}lrrrr@{}}
\toprule
 & \multicolumn{2}{c}{CedarBench (221)} & \multicolumn{2}{c}{Real-world (3)} \\
\cmidrule(lr){2-3}\cmidrule(lr){4-5}
Method & Pass@1 & Pass@3 & Pass@1 & Pass@3 \\
\midrule
Direct, Haiku 4.5 & 0/221   & 0/221   & 0/3 & 0/3 \\
Direct, GPT-5.5   & 155/221 & 199/221 & 0/3 & 0/3 \\
\midrule
\sys{} & \multicolumn{2}{c}{\textbf{221/221}} & \multicolumn{2}{c}{\textbf{3/3}} \\
\bottomrule
\end{tabular*}
\end{table}
\subsubsection{Artifact Validation}
\label{sec:artifact-validation}
We use \texttt{cedar validate} to determine whether the generated schemas and policies parse, type-check, and conform to one another (Table~\ref{tab:validation-results}). Direct GPT-5.5 produces a valid artifact for 155 CedarBench tasks in one attempt and for 199 within three attempts, whereas Haiku~4.5 produces none. Neither direct baseline validates on the three real-world corpora. Haiku~4.5 fails mostly at schema parsing, while GPT-5.5 has fewer but more varied schema and policy errors and fails all \texttt{tags} tasks and most \texttt{streaming} tasks. In contrast, every final \sys{} artifact validates, in every CedarBench domain and on the real-world corpora. Validation shows only that the artifacts are well formed, so the next experiments test whether they are correct.

\begin{table}[t]
\centering
\caption{Request results on the three real-world corpora. Both methods use GPT-6-Luna. A request is bound if the generated schema can express it. Coverage is the
share of requests that are bound, bound accuracy is decision accuracy among bound
requests, and end-to-end accuracy counts unbound requests as failures. The better
value of each pair is in bold.}
\label{tab:realworld-results}
\label{tab:semantic-requests}
\footnotesize
\setlength{\tabcolsep}{1.5pt}
\renewcommand{\arraystretch}{1.1}
\begin{tabular*}{\columnwidth}{@{\extracolsep{\fill}}llrrrrr@{}}
\toprule
 & & \multicolumn{2}{c}{Requests} & \multicolumn{3}{c}{Percent} \\
\cmidrule(lr){3-4}\cmidrule(l){5-7}
Corpus & Method & Bound & Correct & Coverage & Bound & End-to-end \\
\midrule
IBM (67)        & Direct & 47  & 42  & 70.15 & 89.36 & 62.69 \\
                & \sys{} & 58  & 58  & \textbf{86.57} & \textbf{100.00} & \textbf{86.57} \\
\midrule
CyberChair (41) & Direct & 17  & 16  & 41.46 & 94.12 & 39.02 \\
                & \sys{} & 31  & 31  & \textbf{75.61} & \textbf{100.00} & \textbf{75.61} \\
\midrule
iTrust (208)    & Direct & 189 & 51  & \textbf{90.87} & 26.98 & 24.52 \\
                & \sys{} & 181 & 179 & 87.02 & \textbf{98.90} & \textbf{86.06} \\
\bottomrule
\end{tabular*}
\end{table}

\subsubsection{Semantic Correctness}

We assess semantic correctness through real-world authorization requests and comparison with independently compiled gold policies.

For the request experiment, we derive schema-independent natural-language requests and expected decisions from RAGent’s rule annotations of the source requirements~\cite{jayasundara2024ragent}.
This yields 67 IBM, 41 CyberChair, and 208 iTrust requests. Under each method's generated schema, we manually encode each intent as a valid Cedar request, create the supporting entities and context, and evaluate it against the corresponding policy. We count an intent as bound only if its principal, action, resource, and relevant conditions can be represented in the generated schema without changing the intended request. Because the RAGent classifications include some workflow-oriented statements and the underlying requirements can be context-dependent or inconsistent, coverage reflects both schema expressiveness and source-data quality.

Table~\ref{tab:realworld-results} shows that, among requests representable under the generated schemas, \sys{} makes more accurate authorization decisions than direct generation on all three corpora. It correctly decides every bound IBM and CyberChair request and 179 of 181 bound iTrust requests. Direct generation binds slightly more iTrust requests because its schema uses generic resources and broad actions, making more intents representable but preserving fewer authorization-specific distinctions. Consequently, it correctly decides only 51 of 189 bound iTrust requests. The two remaining \sys{} mismatches involve authorization relationships expressed only indirectly in the source requirements.

To complement these request-level results, we evaluate the 200-case gold-reference dataset. This experiment isolates policy generation rather than schema synthesis: both approaches receive the same natural-language requirement and independently compiled gold schema, while the gold policy is withheld. Both methods use GPT-6-Luna at medium effort. A candidate counts as gold-equivalent only if it validates, passes runtime-safety checks, and permits exactly the same requests as gold for every action in the case under the supplied schema. \texttt{cedar symcc} checks equivalence through mutual inclusion of the generated and gold permit sets. We also evaluate the generated verification plan: for each case, let $G$ denote the gold permit set and $L$ and $U$ the plan's floor and ceiling. We check whether gold fits within the plan, $L \subseteq G \subseteq U$.

\begin{table}[t]
\centering
\caption{Policy generation on the 200-case gold-reference dataset. Both methods use
GPT-6-Luna.}
\label{tab:gold-fidelity}
\footnotesize
\setlength{\tabcolsep}{3pt}
\renewcommand{\arraystretch}{1.1}
\begin{tabular*}{\columnwidth}{@{\extracolsep{\fill}}lrr@{}}
\toprule
Method & Completed & Gold-equivalent \\
\midrule
Direct generation & \textbf{200/200 (100.0\%)} & 132/200 (66.0\%) \\
\sys{}            & 197/200 (98.5\%) & \textbf{178/200 (89.0\%)} \\
\bottomrule
\end{tabular*}
\end{table}


Table~\ref{tab:gold-fidelity} shows that \sys{} achieves gold equivalence in 178/200 cases (89.0\%), compared with 132/200 (66.0\%) for direct generation, a 23-point improvement. All 197 completed \sys{} policies pass their generated verification plans, and 178 are also gold-equivalent. The other 19 completed policies satisfy their generated plans but differ from gold because gold does not fit within the plan. The floor includes requests that gold denies (16 cases), the ceiling excludes requests that gold allows (2), or both (1). Three runs do not complete. 

\subsubsection{Human Alignment}
This experiment asks whether human readers prefer the policy artifacts produced by
\sys{} or by direct LLM generation when they must answer concrete
access-control questions. We ran a blinded artifact-judgment study with 14
participants. Participants saw requirement context, an access request, and Cedar
excerpts from two anonymized systems. The A/B labels were randomized per probe.
For each probe, participants selected which artifact better handled the request,
or chose tie/neither, and could explain their judgment.

The study produced 98 probe-level answers across two scenarios: iTrust
credentials, sessions, and patient scope; and gradebook external-grade
conflicts. Participants selected the \sys{} artifact in 82/98 answers (83.7\%),
the direct LLM baseline in 8/98 answers (8.2\%), and tie/neither in 8/98 answers
(8.2\%). The preference was strongest on iTrust, where \sys{} was selected in
52/56 answers (92.9\%). Gradebook was more ambiguous but still favored
\sys{}: 30/42 answers (71.4\%) selected the \sys{} artifact.

Participants' explanations name the same faults as the request experiment.
Participants preferred artifacts that made the relevant boundary explicit:
sensitive credential fields, session state, represented-patient scope, grade
kind, and role-conflict forbids. They criticized baseline artifacts when they
used broad permits or lacked the schema fields needed to enforce the intended
boundary. This study is not a proof of policy correctness, and it does not yet
evaluate \sys{} in human review (Section~\ref{sec:framework}); optional feedback from
human reviewers was sparse.

\subsubsection{Answering RQ3+RQ4}
\textbf{Thus, satisfying the verification plan is a strong indicator of correctness. Of the 197 policies that satisfy their plans, 178 (90.4\%) match gold, and each of the other 19 errors comes from the plan.}
\textbf{With the same model, \sys{} decides bound requests more accurately on all three corpora, matches gold in 89.0\% of cases versus 66.0\%, and is preferred by readers in 82 of 98 answers.}
\section{Threats to Validity}
\label{sec:threats}

\textbf{Target validity.} \sys{} certifies a policy against the reviewed
Cedar schema and verification target, not against organizational truth in the
abstract. If a requirement is missing from the source material, or if a reviewer
approves the wrong atom, the verifier will faithfully check the wrong boundary.
In our experiments, the CedarBench plans are benchmark artifacts. On the requirements corpora, we derived requests and their expected decisions from RAGent’s~\cite{jayasundara2024ragent} access-control sentence classifications and rule annotations, without consulting either system’s output. These annotations come from a third party, but turning each annotated rule into a concrete request still involved our judgment. RAGent’s classification also includes some workflow-oriented sentences, which lowers schema coverage for both methods. The gold-reference prose was rendered by an LLM from structured specifications over domain skeletons that we designed, so it is more explicit than real requirements documents.

\textbf{Model and reviewer effects.} LLM outputs are nondeterministic, and
direct-generation baselines are especially sensitive to retries; this is why we
report retry-based validation results rather than a single sample. The human-review
parts of the evaluation also depend on the reviewer population and review
procedure. A small or author-driven review study can show whether the artifacts
are inspectable, but it cannot by itself prove how production engineers would
review the same targets under deployment pressure.

\section{Related Work}
\label{sec:related}

\textbf{CEGIS when the proposer is stochastic.} Counterexample-guided
inductive synthesis pairs generation with a verifier that rejects bad
candidates~\cite{solarlezama2008,abate2018,jha2023cegis,li2024enumerative}. In classical settings, the
generator is often symbolic or enumerative, so earlier counterexamples can
be removed from the search space by construction. An LLM breaks that
assumption. It can read counterexamples and try again, but it does not
maintain a symbolic version space and can reintroduce an old failure in the
next sample. Recent LLM systems still use this verifier-in-the-loop pattern
for repair, execution-guided debugging, formal proofs, and
planning~\cite{chen2024selfdebug,olausson2024selfrepair,polu2020,jiang2023,first2023,wu2024lemur,varambally2025,kung2026leap}.
\sys{} keeps the useful discipline of CEGIS--every accepted candidate is
checked against the full target--while making explicit what the LLM does
not provide: monotone elimination of prior failures.

\textbf{Verifier-guided formalization.} Formal mathematics systems show how
much an LLM can gain from a mechanically checked environment:
language-model provers have been paired with automated theorem provers,
interactive proof environments, retrieval over verified libraries, proof
repair, and proof-assistant feedback during search or training
~\cite{jiang2022thor,yang2023leandojo,jiang2023,first2023,xin2024deepseekprover15}.
LEAP~\cite{kung2026leap} combines informal proof sketches, Lean compiler
feedback, and a graph that records proof decomposition and memoizes verified
subgoals. Roundtrip verification work makes the neighboring point for
autoformalization: the formal tool can expose semantic drift and support scoped
repair~\cite{amrollahi2026}.
Like work that turns intent into formal
specifications~\cite{cosler2023nl2spec,endres2024nl2postcond,lahiri2024intent}, \sys{} starts from organizational intent and
must first construct the object to be checked: a reviewed schema, verification
plan, and a trace that links source sentences to atoms, checks, and policy rules.

\textbf{LLM policy authoring.} Prior work asks whether LLMs can synthesize
access-control policies from examples, request specifications, or natural
language~\cite{vatsa2025synth,jayasundara2024sok,jayasundara2024ragent}, and whether they can summarize or
reconstruct existing policies at scale~\cite{vatsa2025explore}. Those
studies expose the failure mode that motivates this paper: syntax and
plausibility are easier than semantic equivalence to the intended request
set. \sys{} therefore moves the question from ``can the model write a
policy?'' to ``can it search under a reviewed target checked by a
verifier?''
Closest to \sys{}, AgentCore Policy~\cite{agentcore2026} and Mondl et al.~\cite{mondl2026autoformalization} translate natural language into Cedar and check the result with Cedar's validator and symbolic analysis, but against generic properties (validity, vacuity, conflicts, or an LLM judge) rather than a reviewed per-requirement target; Kiro's requirements analysis~\cite{kiro2026requirements} formalizes requirements to expose ambiguity but synthesizes no policy.

\textbf{Symbolic analysis as oracle, not repair.} Policy analyzers such
as Zelkova use SMT solving to reason about cloud policies~\cite{zelkova2018},
and related systems quantify or repair permissiveness in access-control
policies~\cite{eiers2022quantifying,eiers2023repair}. These tools can
produce the semantic facts we need: inclusion results, equivalence results,
and witnesses. A stochastic proposer needs something different from the raw
oracle output. It needs to know whether a witness means excess permission,
missing permission, or a vanished workflow. \sys{} focuses on that missing
interface between symbolic policy evidence and model-facing repair.
\section{Conclusion}
\label{sec:conclusion}

An LLM that writes a plausible policy has not earned authority over access
control. \sys{} assigns it a narrower role: propose candidates after a
reviewed target has been fixed, and let a verifier judge every proposal
against that target. \sys{}’s guarantee is intentionally narrow: a returned policy is certified against the approved plan, not against intent itself. Against independently compiled gold policies, \sys{}’s policy matched the gold policy in 178 of 200 cases, versus 132 for direct generation with the same model. Every \sys{} miss among completed runs traced to its plan rather than to synthesis, which is where review belongs.

\textbf{Use of generative AI.} We used LLM assistants to draft and edit portions of the text and to check consistency between the text and the tables; the authors reviewed all content and take full responsibility for it. LLMs used as components of \sys{} and to construct the evaluation datasets are described in Sections~\ref{sec:framework} and~\ref{sec:evaluation}.

\bibliographystyle{IEEEtran}
\bibliography{references}

\end{document}